\documentclass[conference]{IEEEtran}
\IEEEoverridecommandlockouts
\usepackage{cite}
\usepackage{graphicx}
\usepackage{textcomp}
\usepackage{xcolor}
\usepackage{amsmath}
\usepackage{amsthm}
\usepackage{amssymb}
\usepackage{hyperref}
\usepackage[final]{pdfpages}
\usepackage[all]{hypcap}
\usepackage{amsmath}
\usepackage{algpseudocode}
\usepackage{tikz}
\usepackage{pgfplots}
\usepackage{subcaption}
\usepackage{mathtools}

\algnewcommand\algorithmicforeach{\textbf{for each}}
\algdef{S}[FOR]{ForEach}[1]{\algorithmicforeach\ #1\ \algorithmicdo}

\usepackage{mathtools}

\newtheorem{theorem}{Theorem}[section]

\usepackage[linesnumbered,ruled, vlined,boxed,commentsnumbered]{algorithm2e}
\usepackage{comment}

\usepackage{lipsum}

\graphicspath{{fig/}}

\def\BibTeX{{\rm B\kern-.05em{\sc i\kern-.025em b}\kern-.08em
    T\kern-.1667em\lower.7ex\hbox{E}\kern-.125emX}}
\begin{document}

\title{Off-Chain Micropayment Pool for High-Throughput Bandwidth Sharing Rewards}

\author{\IEEEauthorblockN{Jieyi Long}
\IEEEauthorblockA{\textit{Theta Labs, Inc}\\
San Jose, CA, USA}
\and
\IEEEauthorblockN{Ribao Wei}
\IEEEauthorblockA{\textit{Theta Labs, Inc}\\
San Jose, CA, USA}
}

\maketitle

\begin{abstract}
This paper presents a layer-2 micropayment pool design which supports high-throughput blockchain off-chain payment, designed specifically for the use case of rewarding users who share redundant bandwidth to assist video stream delivery. We analyze the validity and effectiveness of the proposed micropayment pool. Our analysis results demonstrate that the proposed micropayment pool design is better suited for the bandwidth reward use case compared to existing off-chain payment channel solution.
\end{abstract}


\section{Introduction}
Internet video accounts for over three-quarters of all internet traffic today, and will increase further to 82\% by 2022, according to Cisco’s global IP traffic forecast \cite{cisco2017iptraffic}.  The same report predicts that from 2017 to 2022, global Internet video traffic will grow four-fold, live Internet video will grow 15-fold, virtual reality and augmented reality traffic will grow 12-fold.  In the U.S., video consumption among the millennials has jumped 256\% from an average of 1.6 hours per week to 5.7 hours per week according to a SSRS Media and Technology survey, and mobile devices are leading the charge in video consumption.

Content Delivery Networks (CDNs), which are systems of distributed servers that deliver content data to end users, are predicted by Cisco to carry 72\% of Internet traffic by 2022. They play an important role in distributing web content and streaming video data, by providing a backbone infrastructure for the data delivery.  A major limitation of today’s CDN networks is the so-called ``last-mile'' delivery problem, where a last-mile link between a limited number of data centers and end users presents a bottleneck in the data streaming and delivery pipeline and often leads to less optimal user experience, including link failures, noticeable delays, choppy streams, poor picture quality, and frequent rebuffering.  Another major concern is the CDN bandwidth cost, which can easily reach millions of dollars per year for popular video streaming sites.  These issues become more prominent with the coming era of high resolution digital media, for example 4K, 8K, and 360-degree VR streaming, and upcoming technologies such as light field streaming.

To overcome such bandwidth limitations, decentralized peer-to-peer data streaming and delivery platforms have been developed based on self-organizing and self-configuring mesh networks \cite{ipfs,peer5,thetanetwork,pouwelse2005bittorrent}.  These solutions typically enhance the video player with a peer-to-peer stream delivery SDK. While watching the video streams, end users may \textbf{opt-in} to share redundant or unused bandwidth resources via the SDK, and thus greatly reduces the bandwidth pressure of the CDN servers, resulting in higher stream delivery quality.  However, in practice, relatively few users will voluntarily share their resource unless properly compensated. Without an incentive mechanism in place, we may not expect many users to get on board and join the sharing ecosystem. 

%Motivating and incentivizing users to actively share available resources require a secure and minimally delayed award system or payment method that is compatible with the decentralized nature of the peer-to-peer network.

The emerging blockchain platforms \cite{nakamoto2008bitcoin,wood2014ethereum,gilad2017algorand} provide a promising infrastructure for delivering the incentives through cryptocurrency transactions. For example, a viewer who received a video chunk can pay a smaller amount of cryptocurrency to the peer that delivered the chunk, in return for the bandwidth shared. Such a \textbf{pay-per-byte} scheme offers the finest granularity of reward to the peers, and can thus effectively incentivize the bandwidth sharing.  However, to make such a blockchain-based rewarding mechanism practical, the blockchain needs to support ultra-high transaction throughput. For example, for one single video stream with a moderate 100,000 concurrent viewers, assuming each viewer receives and pays for a video chunk from a peer node once every 10 seconds, then the required transaction throughput is already as high as 10,000 transactions per second, exceeding the capability of most of the leading blockchain platforms.

Off-chain/Layer-2 solutions including payment channels \cite{poon2015lightning,decker2015duplexchannels,green2017bolt,avarikioti2018algochannel,avarikioti2019asyncstatechannel,lind2019teechain,miller2019spritesandchannels,avarikioti2020cerberuschanels}, side-chains \cite{back2014pggedsidechains,gazi2019possidechains}, and roll-up \cite{berrywhitehat2018rollup,buterin2018rollup,long2018rolluperc20} provide promising possibilities for blockchain throughput scaling. By moving the majority of transactions off-chain, the system is able to support frequent payments for small, individual video data chunks transmitted to or received from one or more peers. However, as will be discussed in more details later in the paper, these existing solutions may not be the best fit for delivering data sharing rewards. In particular, existing payment channels can only handle a pair of users. However, when watching a video stream, typically a viewer is exchanging video chunks with and need to pay to multiple peers. Setting up a payment channel with each individual peer is highly inefficient. Instead, we need an off-chain mechanism that can handle one-to-many payments.

\subsubsection{Our Contributions.} In this paper, we address this problem with a novel micropayment pool design which is capable of handling one-to-many off-chain payments. Our contributions can be summarized as follows:

\begin{itemize}
  \item We analyze the requirements for a micropayment system for data sharing rewards. 
  \item We propose a micropayment pool design which can handle off-chain one-to-many payments, well suited for delivering the data sharing rewards with pay-per-byte granularity.
  \item We analyze the validity and effectiveness of the micropayment pool, and show that it is more efficient for data sharing rewards than payment channels. 
\end{itemize}

\subsubsection{Related Works.}

Off-chain/Layer-2 payment has been a hot research topic in recent years as a scalability solution for blockchain platforms \cite{poon2015lightning,decker2015duplexchannels,green2017bolt,avarikioti2018algochannel,avarikioti2019asyncstatechannel,lind2019teechain,miller2019spritesandchannels,avarikioti2020cerberuschanels,pass2015micropayment}. A payment channel (or more generally state channel) is a mechanism that allows users to exchange multiple transactions without committing to the blockchain \cite{poon2015lightning,decker2015duplexchannels}.  Such “off-line” transactions can be settled on the blockchain later, thus incurring minimal transaction confirmation latencies.  A payment channel typically start with an opening funding transaction to lock the initial fund onto the blockchain \cite{poon2015lightning}. Then, subsequent commitment transactions are exchanged off-line between the two users to update the initial state.  A final settlement transaction unilaterally or bilaterally closes the payment channel when submitted to the blockchain for confirmation.  As state updates or commitment transactions can be exchanged between the users off-line, as soon as they are created and signed, many more transactions can be exchanged in between the funding and the settlement transactions.  Cutting the number of on-chain transactions down to two also drastically reduces the costs associated with very frequent micropayments. Earlier development revolves around Bitcoin, and evolved into networks of payment channel called the Lightning Network \cite{poon2015lightning}. Over the years, more designs of off-chain payment channels built on top on other blockchain infrastructes, including the Ethereum smart contract, are proposed \cite{raiden2017}. These payment channels typically only handle back-and-forth payments among two parties. The Lightning Network and the alike can potentially handling multi-party payment, but has rather complex setup and requires pre-existing payment routes among the users participated in the value exchange \cite{poon2015lightning,miller2019spritesandchannels,raiden2017}.

Besides payment/state channel, other layer-2 solutions like side-chain \cite{back2014pggedsidechains,gazi2019possidechains} and roll-up \cite{berrywhitehat2018rollup,buterin2018rollup,long2018rolluperc20} are also widely discussed among the blockchain research community. It is worth noting that our off-chain micropayment pool can work on top of these solutions to further extend the throughput of the blockchain.

\section{Motivation}

In this section, we will first review the peer-to-peer video bandwidth sharing problem, which provides the context for our discussion. Then, we will dive into the bandwidth sharing reward problem and discuss the shortcomings of the existing payment channel solutions.

\subsection{Peer-to-Peer Video Bandwidth Sharing} \label{sec:peer_to_peer_bandwidth_sharing}

\begin{figure}[htbp]
\centerline{\includegraphics[width=0.5\textwidth]{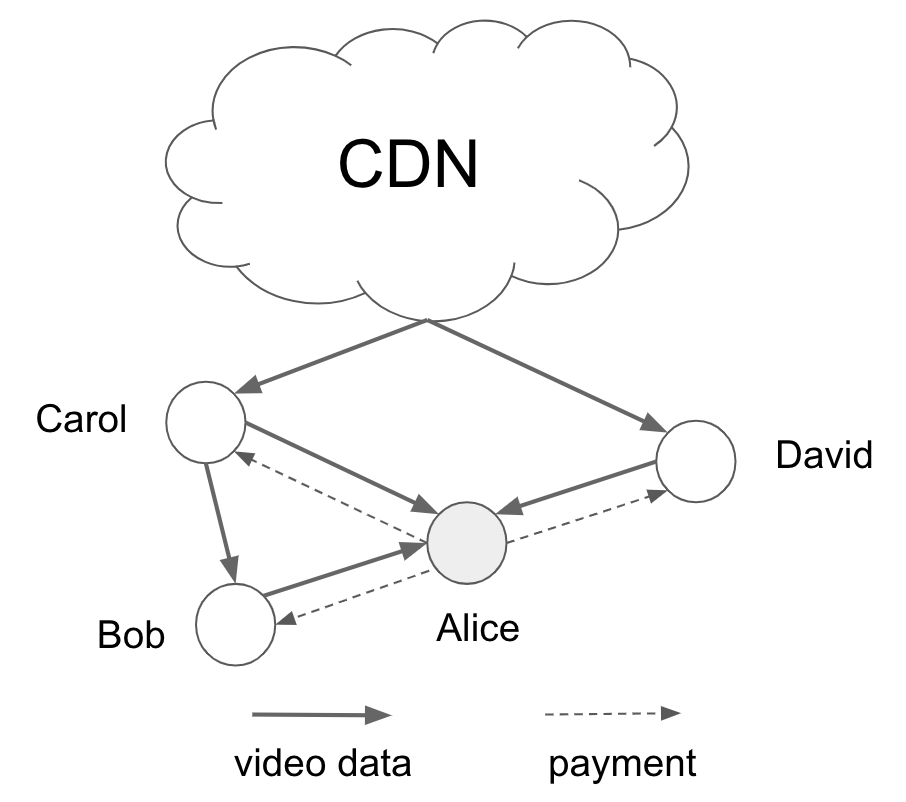}}
\caption{Viewer Alice pulling stream data from multiple peesrs. The solid arrows represent the video stream data flows, and the dashed arrows represent the micropayments between the peers.}
\label{fig:peer_to_peer_stream_sharing}
\end{figure}

Popular videos or live streams, such as NBA final games, could attract millions of people watching concurrently.  Given the sheer number of concurrent users, for any viewer, it is highly likely that there are other viewers nearby enjoying the same video stream. Since they are watching the same content, instead of loading the stream from the CDN, some of them can pull from other peers close by. With sufficient number of concurrent viewers, such a peer-to-peer sharing mechanism can offload significant amount of CDN data. To achieve the maximum offload, a viewer typically pulls the stream from multiple peers, as shown in Figure \ref{fig:peer_to_peer_stream_sharing}.

In the example, Viewer Alice has multiple peers (Bob, Carol, and David, etc) in her proximity. Popular video streaming protocols like HLS/MPEG-DASH typically divides the video file into small video \textbf{data chunks}. A viewer client would download these data chunks in serial order. As the video stream plays, whenever Alice needs to download the next data chunk, the peer-to-peer sharing SDK asks around and download the chunk from peers if one of the peer nodes already has the chunk available locally. The CDN is used as a fallback source when none of the peers has the required chunk. It is reported that such a peer network can potentially result in 50\% to 90\% CDN bandwidth saving \cite{peer5,thetanetwork} if all the viewers participate in stream data sharing.

\subsection{Off-Chain Payment Channels for Bandwidth Sharing Rewards} \label{sec:payment_channel_for_bandwidth_sharing_rewards}

In the example shown in Figure \ref{fig:peer_to_peer_stream_sharing}, viewer Alice can send a small amount of crytocurrency with an on-chain transaction for each video data trunk she obtain from peer nodes. However, as mentioned in the introduction, this could quickly overwhelm the capacity of the blockchain.

An immediate optimization is to adopt the payment channel technique. Indeed, a payment channel can significantly improve the scalability since it reduces the number of on-chain transactions. With a payment channel, a viewer can pay for many video segments with just one on-chain settlement transaction. Nonetheless, there are two issues with the payment channel solution.

First, \textbf{slow node switching} are unconducive to streaming data segments from many cachers to a single viewer.  As discussed above, an on-chain transaction is needed to establish a payment channel between any two parties, which might take at least a couple seconds to be confirmed on the blockchain.  Typically, to make peer-to-peer sharing effective, a viewer is typically connected to 10+ peers to increase the likelihood that at least one of the connected peers possesses the desired video chunk.  Each time the viewer needs to pull a stream from a new peer, it needs to make an on-chain transaction first to open the payment channel, which is time consuming.  In addition, before the on-chain transaction is confirmed, the viewer node cannot pull streams from the peer node. This could halt the video stream, leading to degraded user experience.

Second, the payment channel solution still requires \textbf{a relatively large number of on-chain transactions}.  Each payment channel requires an upfront token lockup transaction, and a settlement transaction to close the channel. With 10+ peers, this implies each viewer is associated with at least 20 on-chain transactions, which still hinders the scalability of the system.

Thus, there is an unsolved need to design an off-chain payment system suited for decentralized video streaming, with support for \textbf{one-to-many} off-chain transactions.  Such a system can effectively address the two issues mentioned above. First, with the one-to-many payment capability, node switching can be done instantly since no intermediate on-chain transaction is required when connecting to a new peer. Second, the one-to-many capability also potentially reduces the number of on-chain transactions needed, and thus increasing the scalability of the system. Of course, the system should be \textbf{trustless}, and able to detect, prevent, and penalize fraudulent activities such as double spending. 

\section{Off-Chain Micropayment Pool}

In this section, we present a ``resource-oriented off-chain micropayment pool'' design that has the desired properties listed in the previous section.

\begin{figure*}[t]
    \centering
    \includegraphics[width=0.9\textwidth]{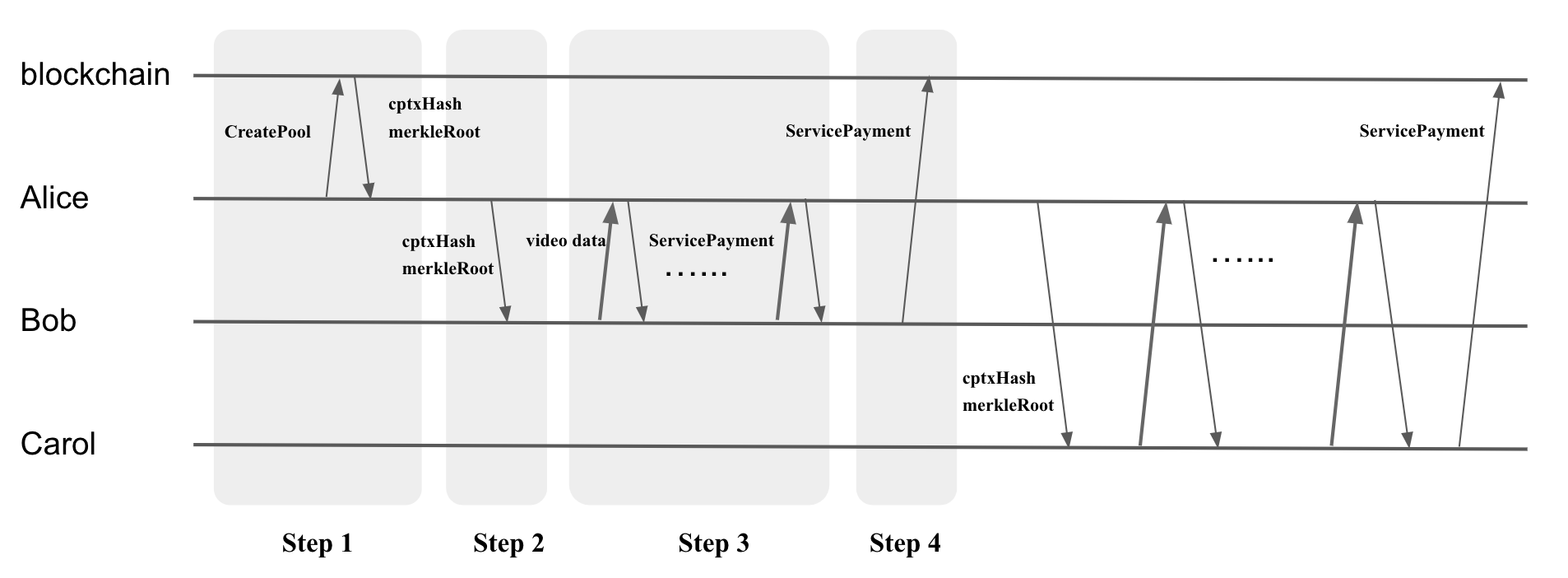}
    \caption{The working mechanism of the resource-oriented off-chain micropayment pool. Here Alice first downloads video chunks from Bob, and pay to Bob via the micropayment pool with the four-step interaction. Then, Alice switches and downloads video chunks from Carol without an intermediate on-chain transaction.}
    \label{fig:resource_oriented_off_chain_micropayment_pool_design}
\end{figure*}

\subsubsection{Resource-Oriented Off-Chain Micropayment Pool.}

Figure \ref{fig:resource_oriented_off_chain_micropayment_pool_design} is a diagram illustrating the working mechanism of our proposed resource-oriented off-chain micropayment pool, showing a viewer Alice making off-chain payments to cachers Bob and Carol for video chunks.  In the following, as we discuss the micropayment pool design, we will keep the discussion at the protocol abstraction level. However, we note that the proposed micropayment pool can be implemented on any Turing-complete smart contract platforms such as Ethereum \cite{wood2014ethereum}. The interaction between Alice and any of her peers (e.g. Bob) can be devided into four steps, as illustrated below.

\bigbreak

\noindent \textbf{Step 1. Micropayment pool creation:} Alice publishes an on-chain transaction to create a micropayment pool with a time-lock $duration$:

\begin{equation} \label{eq:step_1_create_pool}
\textbf{CreatePool}(resourceId, deposit, collateral, duration)
\end{equation}

To create the pool, Alice may specify a \textbf{“Resource ID”} $resourceId$ that uniquely represents the digital content she intends to retrieve. It may refer to a video file, or a live stream.  The $deposit$ amount may be at least the total value of the resource to be retrieved. For instance, if the resource is a video file which is worth 100 tokens, then the deposit has to be at least 100 tokens. Along with the $deposit$, Alice is required to put down a $collateral$ which can discourage Alice from double spending. If a double spending attempt from Alice is detected, the $collateral$ will be slashed.  We will discuss the double spending detection and slashing later in this section. There we will also show that if $collateral > deposit$,  the net return of a double spend is \textit{always negative}, and hence any rational user will have no incentive to double spend.  The $duration$ acts as a time-lock similar to that of a standard payment channel. Any withdrawal (i.e. Step 4. On-chain settlement) from the micropayment pool has to be before the time-lock expires.  The blockchain returns Alice the Merkle proof of the CreatePool transaction after it has been committed to the blockchain, as well as $cptxHash$, the transaction hash of the CreatePool transaction.

\bigbreak

\noindent \textbf{Step 2. Initial handshake between peers:} Whenever Alice wants to retrieve the specified resource from a peer (Bob, Carol, or David, etc.), she sends the Merkle proof and $cptxHash$, the transaction hash of the on-chain CreatePool transaction to that peer. The peer verifies the Merkle proof to ensure that the pool has sufficient deposit and collateral for the requested resource, and both parties can proceed to the next steps. Since Alice already has the required Merkle proofs in hand, and the peer can verify the proofs instantly with the corresponding block headers, this handshaking can be done instantly.

\bigbreak

\noindent \textbf{Step 3. Off-chain payments:} Alice signs ServicePayment transactions and sends them to the peers off-chain in exchange for parts of the specified resource (e.g., a piece of the video file, a live stream segment, etc.). A ServicePayment transaction may contain the following data:

\begin{equation} \label{eq:step_3_serice_payment}
\textbf{ServicePayment}(target, amount, cptxHash, tgtSeq, \sigma)
\end{equation}

\noindent Here $target$ is the address of the peer that Alice retrieves the resource from, and $amount$ is the amount of token payment Alice intends to send. Similar to payment channels, the $amount$ should increment each time to reflect the total value of the resource Alice already received since the handshake (Step 2). In the video streaming context, say each video chunk is worth 5 token. Then the first $ServicePayment$ Alice sent to Bob should set $amount$ to 5, and the second $ServicePayment$ should set $amount$ to 10, and so on. $tgtSeq$ is to prevent a replay attack.  It is similar to the ``nonce'' parameter in an Ethereum transaction. If a target publishes a ServicePayment transaction to the blockchain (see Step 4), its $tgtSeq$ needs to increment by one. Alice also attachs her signature $\sigma$ to the ServicePayment transaction, where $\sigma$ is computed by

\begin{equation}\label{eq:step_3_serice_payment_signature}
\sigma = \textbf{Sign}(sk_{Alice}, target || amount || cptxHash || tgtSeq))
\end{equation}

\noindent As the peer (Bob, Carlo, or David, etc.) receives the off-chain ServicePayment transactions, it needs to verify validity of the transactions, including the $amount$, $cptxHash$, $tgtSeq$, and the signature $\sigma$. Upon validation, the peer can send Alice the requested data chunks of the resource specified by the CreatePool transaction. Note that the off-chain ServicePayment transactions can be sent directly between two peers.  Hence there is no scalability bottleneck for this step. On the other hand, to ensure that the data returned by the peers are parts of the requested resource, Alice can compare the hashes of the parts against the expected part hashes, which are small strings downloadable form the CDN.

\bigbreak

\noindent \textbf{Step 4. On-chain settlement:} Any peer (Bob, Carol, or David, etc.) that receives the ServicePayment transactions from Alice can sign and publish the transactions to the blockchain anytime before the timelock expires to withdraw the tokens.  ServicePayment transactions that are published on-chain may also be called ``on-chain settlement'' transactions. The protocol also requires the peer to pay for the transaction fee for the on-chain settlement.

\bigbreak

Finally, we note that if Bob cannot provide the video chunks Alice requested, Alice may switch to another peer (e.g. Carol). As illustrated by Figure \ref{fig:resource_oriented_off_chain_micropayment_pool_design}, when the switch happens, no on-chain transaction is needed. This means \textbf{the viewer can switch to any peer node at any time without making an on-chain transaction} that could potentially block or delay the video stream delivery. We also note that for the purpose of illustration, in Figure \ref{fig:resource_oriented_off_chain_micropayment_pool_design} Carol starts to interact with Alice only after Bob submits the ServicePayment transaction to the blockchain. In practice though, multiple peers can interact with Alice in parallel, and the micropayment pool still works correctly (Theorem \ref{thm:double_spending_resistance}).

%Moreover, the total number of tokens needed to create the micropayment pool is $(collateral + deposit)$, which can be as low as twice of the value of the requested resource, no matter how many peers Alice retrieves the resource from. Thus, the amount of reserved token reduces from $O(n)$ to $O(1)$ compared to the unidirectional payment channel approach, where $n$ is the number of peers Alice retrieves the resource from.

\subsubsection{Double Spending Detection and Penalty.} 

A malicious actor may attempt to make a double spending if it is profitable. Thus, the micropayment pool needs to equipped with a proper mechanism to penalize double spending attempts.  To detecting double spending, the on-chain smart contract that implements the micropayment pool should check every on-chain transaction. The \textbf{double spending detection rule} is simple: if a remaining deposit in the micropayment pool cannot cover the next consolidated payment transaction signed by both Alice and another peer, the smart contract will consider that Alice has conducted a double spend.  The smart contract should also \textbf{slash the collateral} of the micropayment pool as the penalty to Alice. We will prove in Theorem \ref{thm:double_spending_resistance} that this  ensures that Alice is worse off if she double spends.

\section{Analysis}

\subsubsection{Assumption.} Recall that in our context, a \textbf{resource} refers to a piece of digital content. We can thus make the following self-evident assumption before proving the correctness of the proposed resource-oriented micropayment pool:  Possessing multiple copies of a resource does not increase the value the resource. For example, if Alice already downloads a video file from Bob, then the same video downloaded from Carol does not provide any extra value to Alice. 

\subsubsection{Correctness Analysis.} In this section, we will prove the correctness of the proposed micropayment pool under the above assumption.

\begin{theorem}\label{thm:double_spending_resistance} The creator of the resource-oriented micropayment pool cannot profit from double spending. 
\end{theorem} 

\begin{proof} We assume Alice is malicious, while her peers Bob, Carol, and David maybe honest. Even worse, Alice could colludes with another malicious peer Edward. Alice exchanges partially signed transactions with Bob, Carol, and David for the specified resource. Based on our assumption above, Alice gains no extra value for the duplicated resource, the maximum value she gets from Bob, Carol, and David is at most the deposit amount. 

As Alice colludes with Edward, she can send Edward the full deposit amount. She then asks Edward to commit the settlement transaction before anyone else and return her the deposit later. In other words, Alice gets the resource which is worth at most the deposit amount for free, before the double spending is detected. Later when Bob, Carol, or David commit the settlement transaction, the double spending is detected, and the full collateral amount will be slashed. Hence, the net return for Alice is

\begin{equation}\label{eq:double_spend_net_profit}
net_{Alice} = deposit - collateral
\end{equation}

Therefore, for this scenario, as long as $collateral > deposit$, Alice’s net return is negative. Hence, if Alice is rational, she would not have any incentive to double spend.  Similarly, analysis for other cases show that Alice’s net return is always negative if she conducts a double spend.
\end{proof}

\begin{theorem}\label{thm:bounded_loss} A malicious peer might not return the requested data to the micropayment pool creator after receiving the ServicePayment. However, the loss of the pool creator can be bounded by $k \cdot A$, where $k$ is the number of peers, and $A$ is the amount of tokens paid for each data chunk.
\end{theorem}

\begin{proof} First, the pool creator Alice has a way to verify if a data chunk (i.e. a part of the resource) returned by a peer is the requested chunk. Since the CreatePool transaction specifies the $resourceId$, all the users should have a common understanding of which resource is expected to be shared. And as described in ``Step 3. Off-chain payments'', with the $resourceId$, they can download the hashes of the data chucks  from the CDN and verify the data chunks received from the peers.

A malicious peer might not return the requested data to Alice after receiving the ServicePayment. If such behavior is detected, Alice can \textit{blacklist} that peer turn to another peer to get the resource. Thus, the lost of Alice is bounded by $k \cdot A$. Moreover, if the resource is divided into small enough chunks, in theory $A$ can be made infinitesimally small. Then, the total lost $k \cdot A$ will also be infinitesimally small. In a more practical case where some the peers is honest, Alice should be able to identify those peers and get the requested resource from them.
\end{proof}

\begin{theorem} The proposed resource-oriented micropayment pool can incentivize the users to minimize the number of on-chain transactions.
\end{theorem}

\begin{proof}
As specified in ``Step 4. On-chain settlement'', the protocol requires the peer (i.e. the target address of the ServicePayment transaction) to pay for the transaction fee for the on-chain settlement. To pay less transaction fees, a peer would have the incentive to publish on-chain settlements only when necessary, which is beneficial to the scalability of the network.
\end{proof}

\section{Results and Discussions}

We have implemented the proposed resource-oriented micropayment pool in a testnet environment and compare its effectiveness with the traditional payment channels. As also discussed in Section \ref{sec:payment_channel_for_bandwidth_sharing_rewards}, we compare two aspects that are critical for the bandwidth sharing reward use case, i.e. the node switching time and the number of on-chain transactions. The comparison shows that our proposed resource-oriented micropayment pool has multiple advantages for bandwidth sharing rewards compared to existing payment channel solution.

\bigbreak

\noindent \textbf{Node switching time:} As discussed in the previous sections, in the peer-to-peer stream sharing environment, due to high peer churn rate, a viewer needs to connect to new peer frequently. In our testnet environment, such node switching happens once every 60 to 100 seconds on average. With the proposed micropayment pool, after the initial CreatePool transaction, when the viewer pulls the video stream data from a new peer, no on-chain transaction is required. Instead, the viewer can just send the Merkle proofs of the CreatePool transaction to the peer, and the peer can then validate it with the corresponding block header. In our implementation, the entire process can be conducted within 100 milliseconds. Such fast turn-around time means the viewer can pull data from the new peer instantly, and thereby improves the stream delivery quality. In contrast, using traditional payment channel, the stream data exchange with the new peer cannot start until the token locking transaction is confirmed. On Ethereum, this could take up to 90 seconds (6 block confirmation). On newer blockchains like Algorand, transaction confirmation can still takes several to tens of seconds, making the peer-to-peer stream sharing much less efficient.

\bigbreak

\noindent \textbf{Number of on-chain transactions:} For a viewer, our one-to-many micropayment pool only requires one CreatePool transaction upfront. Then for each peer the viewr connects to, only one on-chain settlement transaction is needed to transfer the tokens. Thus, for a viewer with $n$ peer, a total of $n+1$ on-chain transactions are needed. In comparison, using the traditional payment channel, each peer is associated with an opening and a closing on-chain transactions. Thus, a total of $2n$ on-chain transactions are needed. As $n$ is usually larger than 10 in a typical streaming session, given the same blockchain transaction throughput, our proposed micropayment pool can increase the maximum number of concurrent viewers by almost 100\%. Thus, our resource-oriented micropayment pool can significantly improve the scalability of the blockchain for bandwidth sharing rewards.

\bigbreak

%
% ---- Bibliography ----
%

\bibliographystyle{unsrt}
\bibliography{ms}

\end{document}